\documentclass[pdflatex,sn-mathphys-num]{sn-jnl}

\usepackage{graphicx}%
\usepackage{multirow}%
\usepackage{array}
\usepackage{amsmath,amssymb,amsfonts}%
\usepackage{amsthm}%
\usepackage{hyperref}
\usepackage{mathrsfs}%
\usepackage[title]{appendix}
\usepackage{xcolor}%
\usepackage{textcomp}%
\usepackage{manyfoot}%
\usepackage{booktabs}%
\usepackage{algorithm}%
\usepackage{algorithmicx}%
\usepackage{algpseudocode}%
\usepackage{listings}%
\usepackage{anyfontsize}

\theoremstyle{thmstyleone}%
\newtheorem{theorem}{Theorem}
\theoremstyle{thmstyletwo}%

\theoremstyle{thmstylethree}%

\begin{document}

\title{On Modeling Cylindrical Data with a Discrete Circular Component and Its Environmental Applications}

\author[1]{\fnm{Brajesh Kumar} \sur{Dhakad}}\email{brajeshdhakad85047@gmail.com}

\author*[2]{\fnm{Jayant} \sur{Jha}}\email{jayantjha@gmail.com}
\affil[1]{\orgdiv{Department of Mathematics}, \orgname{Indian Institute of Technology Madras}, \orgaddress{\street{IIT P.O., Sardar Patel Road}, \city{Chennai}, \postcode{600036}, \state{Tamil Nadu}, \country{India}}}

\affil*[2]{\orgdiv{Statistical Science Division}, \orgname{Indian Statistical Institute, Kolkata}, \orgaddress{\street{203 B. T. Road}, \city{Kolkata}, \postcode{700108}, \state{West Bengal}, \country{India}}}


 \abstract{Standard statistical methods are often inadequate for modeling the joint dependence between linear and circular variables, and existing methods for modeling this dependence are designed only for continuous variables. However, circular data are frequently observed on a finite set of equally spaced directions, either due to rounding prior to reporting or because of the experimental design employed for data collection. To address this gap, we propose a flexible, analytically tractable model for jointly representing a discrete circular and a continuous linear variable. The construction combines a wrapped symmetric geometric distribution, a Weibull distribution, and a trigonometric linking function. This formulation yields closed-form expressions for the joint, marginal, and conditional distributions. The choice of the Weibull distribution facilitates direct sample generation using the inverse transform technique. Additionally, it provides explicit expressions for conditional moments, enabling a flexible circular--linear regression framework. We detail the theoretical interpretation of the model parameters, mathematically establishing the monotonicity of the conditional mean and variance with respect to the dependence parameters. The performance of the estimators is demonstrated through extensive simulations, and the utility of the model is illustrated by analyzing two empirical environmental datasets.}

\keywords{Bivariate distribution, Cylindrical data, Circular–linear correlation, Circular–linear regression, Marginal and conditional distributions}

\maketitle
 
\section{Introduction}\label{sec1}

Cylindrical data consist of a circular variable and a linear variable, representing phenomena in which both angular and linear components are observed jointly. Such data arise in a wide range of real-world contexts across diverse fields such as environmental science, meteorology, oceanography, biology, physics, image analysis, and astronomy. For example, in the study of atmospheric pollution, it is often of interest to determine whether the concentration of a pollutant depends on the wind direction. Specifically, an example involving ozone concentration and wind direction is presented in \cite{mardia2009} (Example 11.1, p. 246), where wind direction records are discrete. Applications in other fields are also discussed in the introduction of the same text. For an extensive review of these fields, one may also refer to \cite{Mardia1975, Fisher1993, Jammalamadaka2001, SenGupta2022}.

Analyzing the association in cylindrical data is non-trivial because linear and circular variables exhibit fundamentally different characteristics. Linear variables take values along the real line with a natural ordering and infinite range, whereas circular variables are defined on the unit circle, exhibiting periodic behavior with no fixed origin. Consequently, standard bivariate distributions defined on Euclidean space are unsuitable for describing the joint variability of such data. This necessitates specialized statistical models that explicitly account for the directional nature of the circular component.

In addition to these topological differences, a practical challenge often arises from the resolution of the data. Circular data are frequently observed on a finite set of equally spaced directions; specifically, this involves observations that take values in the set of $m$ equally spaced categories around the unit circle numbered $0, 1, \dots, m-1$, where the $i$-th direction corresponds to the angle $2 \pi i/m$. This discreteness may arise from rounding prior to reporting or result directly from the experimental design employed in gathering the data. For instance, meteorological stations, wave buoys, and animal-orientation sensors frequently report directions in a finite set of equally spaced sectors of the unit circle instead of exact angles. Since the discrete circular variable is often paired with a continuous linear variable in practice (e.g., wind direction and speed), this necessitates statistical models explicitly designed for discrete circular and continuous linear variables.

Existing models for cylindrical data are defined exclusively for continuous variables. For instance, \cite{mardia1978} proposed a model for cylindrical variables, which was further extended by \cite{kato2008}. \cite{johnson1978} discussed angular--linear distributions and related regression models. \cite{garcia2013} examined an algorithm for the flexible estimation of the joint density of circular--linear variables, and a flexible model for cylindrical data is discussed in \cite{abe2017}. However, adapting these continuous frameworks for discrete circular variables requires integrating their densities over the fixed intervals corresponding to the circular variables. This process rarely yields a closed-form solution, necessitating the use of numerical methods that introduce both computational burden and approximation error. Therefore, a tractable model for discrete circular and continuous linear variables is needed. 

This article presents a tractable joint model for discrete circular and continuous linear variables. This model is constructed by combining the wrapped symmetric geometric (WSG) distribution, examined in \cite{Jha2018}, and the Weibull distribution with a trigonometric linking function. A distinct advantage of this formulation is that, unlike other standard positive-valued distributions such as the gamma, the Weibull distribution admits a closed-form inverse cumulative distribution function, which simplifies random number generation. Closed-form expressions are derived for the joint density, marginals, conditionals, and circular--linear regression. Furthermore, we provide a detailed interpretation of the model parameters, specifically establishing the monotonicity of the conditional mean and variance with respect to $\gamma, \beta$, and $d$. The proposed method is applied to analyze a wind dataset, which consists of grouped observations of wind direction and speed, and an SO$_2$ Concentration-Acrophase dataset, based on the diurnal behavior of Sulfur Dioxide (SO$_2$) concentrations. 

We compare the proposed model against discretized versions of well-known continuous formulations. These include the JW model proposed by \cite[Eq.~2.1]{johnson1978}, alongside the models introduced by \cite{abe2017}: the Weibull sine-skewed von Mises (WeiSSVM) and the Generalized Gamma sine-skewed von Mises (GGSSVM).
 
The rest of the paper is organized as follows: Section \ref{sec2} introduces the proposed distribution, and Section \ref{sec3} discusses its properties. Further, Sections \ref{sec4} and \ref{sec5} discuss the circular--linear regression and the correlation framework, respectively. Next, Section \ref{sec6} studies the proposed models using simulated data, and Section 
\ref{sec7} presents the applications of the proposed methodology to real datasets. Section \ref{sec8} concludes. 


\section{Bivariate Weibull-Based Wrapped Geometric Model (BWBWG)}\label{sec2}

Let $\mathbb{Z}_{+}$ and $\mathbb{R}^{+}$ denote the sets of positive integers and positive real numbers, respectively, and let $\mathbb{Z}_\mathrm{m}$ be the cyclic group of integers of order $\mathrm{m \in \mathbb{Z}_{+}}$. Define
\begin{equation*}
   \mathcal{D}_{\mathrm{m}}
   = \left\{ \frac{2\pi i}{\mathrm{m}} : i \in \mathbf{Z}_\mathrm{m} \right\},
\end{equation*}
which forms a regular lattice on the unit circle.

The proposed joint mass--density for the random variables $(\Theta, X)$, defined on  
$\mathcal{D}_{\mathrm{m}} \times \mathbb{R}^{+}$, is given by:
\begin{align}\label{eq:1.1}
\scalebox{0.85}{$
 f(\theta,x) \propto \Big(q^\zeta + q^{\mathrm{m}-\zeta}\Big)x^{\gamma-1} \exp\bigg(-(\beta x)^{\gamma}\Big(1 - \tanh(d)\cos\big(\mu - \frac{ 2 \pi \zeta}{m}\big)\Big)^\varepsilon\bigg),
 $}
\end{align}
where, $\theta \in \mathcal{D}_{\mathrm{m}}, ~\zeta = (k - \alpha) \bmod \mathrm{m}, ~\alpha \in \mathbf{Z}_\mathrm{m},~ x \in \mathbb{R}^{+},~ q \in (0, 1), ~ \mu \in [0, 2\pi),~ \beta>0,~ \gamma > 0,~ d \geq 0$ and  $\varepsilon \in \{-1,1\}$. In this construction, we restrict $d \ge 0$, since any negative value of $d$ is equivalent to a positive one under a shift of the parameter $\mu$ by $\pi$. 
Specifically,
$ 
\tanh(d)\cos\!\left(\mu - \tfrac{2\pi \zeta}{\mathrm{m}}\right) 
= \tanh(-d)\cos\!\left(\mu - \pi - \tfrac{2\pi \zeta}{\mathrm{m}}\right).
$ 

Closed-form expression for the proposed joint density in \eqref{eq:1.1} is given by
\begin{align}
\label{eq:2.5}
\resizebox{0.85\textwidth}{!}{$
f(\theta,x) = C_2 \Big(q^\zeta + q^{\mathrm{m}-\zeta}\Big)x^{\gamma-1} \exp\bigg(-(\beta x)^{\gamma}\Big(1 - \tanh(d)\cos\big(\mu - \frac{ 2 \pi \zeta}{\mathrm{m}}\big)\Big)^\varepsilon\bigg) 
$}
\end{align}
where, {\small $ \frac{1}{C_2} = \frac{p_1 + q^\mathrm{m} p_2}{\gamma \beta^{\gamma}},~ p_1 = \sum_{k = 0}^{\mathrm{m}-1} q^{\zeta} \Big /\Big( 1-tanh(d)cos\big(\mu-\frac{2 \pi \zeta}{\mathrm{m}}\big)\Big)^\varepsilon,~ p_2 = \\ \sum_{k = 0}^{m-1} 1\Big/q^{\zeta} \Big(1-tanh(d)cos\big(\mu-\frac{2 \pi \zeta}{\mathrm{m}}\big)\Big)^\varepsilon$ }.

The distribution given in \eqref{eq:2.5} is referred to as Bivariate Weibull-Based Wrapped Geometric (BWBWG) distribution.   


\section{Properties of the Model}\label{sec3}
This section examines key properties of the model, including the influence of parameters on the dependence structure of $(\Theta, X)$, and the behavior of the marginal and conditional distributions at the boundary values of $q$ and $d$. 

\begin{theorem}\label{theo:2.5}
For the BWBWG, the random variables $\Theta$ and $ X$ are independent if and only if $d = 0$. 
\end{theorem}

This theorem establishes that $d=0$ is a necessary and sufficient condition for the independence of $\Theta$ and $X$.

\subsection{Marginal Distribution of \texorpdfstring{$\Theta$}{Theta} }
The marginal distribution of the circular variable $\Theta$  is given by
\begin{align}\label{eq:3.1}
P(\theta) = \frac{q^{\zeta} + q^{\mathrm{m} - \zeta}}{(p_1 + q^\mathrm{m} p_2)\left(1 - \tanh(d)\cos\left(\mu - \frac{2\pi \zeta}{\mathrm{m}}\right)\right)^\varepsilon},
\end{align}
and is obtained from the BWBWG distribution given in \eqref{eq:2.5}. Since the expression is undefined for $q \in \{0,1\}$, the distributions for these boundary cases are derived by evaluating the corresponding limits.
\begin{itemize}
    \item When $q = 0$, the marginal distribution degenerates to a point mass at $\theta = \tfrac{2\pi \alpha}{\mathrm{m}}$.
    \item When $q = 1$, the marginal distribution reduces to the simplified form:
    \begin{align}\label{eq:3.2}
 P(\theta) = \frac{1}{p_3 \big(1 - \tanh(d) \cos\big(\mu-\frac{2 \pi \zeta}{\mathrm{m}}\big)\big)^\varepsilon},   
\end{align}
where $ p_3 = \sum_{k = 0}^{\mathrm{m}-1}\frac{1}{ \big(1 - \tanh(d) \cos \big(\mu-\frac{2 \pi \zeta}{\mathrm{m}}\big)\big)^\varepsilon}$, and the resulting distribution is asymmetric. Moreover, when $d = 0$, \eqref{eq:3.2} reduces to the uniform distribution, i.e. $ P(\theta) = \frac{1}{\mathrm{m}},~ \forall \theta  \in \left\{\frac{2 \pi k}{\mathrm{m}} : k\in \mathbf{Z}_\mathrm{m} \right\} $.
\end{itemize}

\subsection{Marginal Distribution of  \texorpdfstring{$X$}{X} }
The marginal distribution of the linear variable $X$ is given by 
\begin{equation}\label{eq:3.3}
 f_{X}(x) =  C_2 \sum_{k = 0}^{\mathrm{m}-1} \big(q^\zeta + q^{\mathrm{m}-\zeta}\big)x^{\gamma-1} \exp\Big(-(\beta x)^{\gamma}\Big(1 - \tanh(d)\cos\Big(\mu - \frac{ 2 \pi \zeta}{\mathrm{m}}\Big)\Big)^\varepsilon\Big)     
\end{equation}

\subsection{Conditional Distribution of  \texorpdfstring{$\Theta$}{Theta} given \texorpdfstring{$X$}{X}}
The conditional distribution of $\Theta$ given $X = x$ is
\begin{align} \label{eq:3.7}
 P(\theta \mid x) = \frac{1}{p_4}(q^\zeta + q^{\mathrm{m}-\zeta}) \exp\Big(-(\beta x)^{\gamma}\Big(1 - \tanh(d) \cos\Big(\mu-\frac{2 \pi \zeta}{\mathrm{m}}\Big)\Big)^\varepsilon\Big),
\end{align}
 where $p_4 = \sum_{k = 0}^{\mathrm{m}-1} (q^\zeta + q^{\mathrm{m}-\zeta}) \exp\big(-(\beta x)^{\gamma}\big(1- \tanh(d)\cos(\mu-\frac{2 \pi \zeta}{\mathrm{m}})\big)^\varepsilon\big)$.\\ 

 \autoref{fig:3.5} displays the probability mass function (pmf) given in \eqref{eq:3.7}. It shows how $\alpha,~\varepsilon,~d$ and $X$ affect the pmf under the settings $\mathrm{m} = 10$, $\gamma = \beta = q = 0.5$, and $\mu = \tfrac{\pi}{4}$. The subfigures in the first row correspond to $\varepsilon = -1$, while those in the second and third rows to $\varepsilon = 1$. The subfigures in the first, second, and third columns correspond to $x = 25$, $x = 50$, and $x = 100$, respectively.

\begin{figure}[!ht]
    \centering
    \includegraphics[width= 1\linewidth]{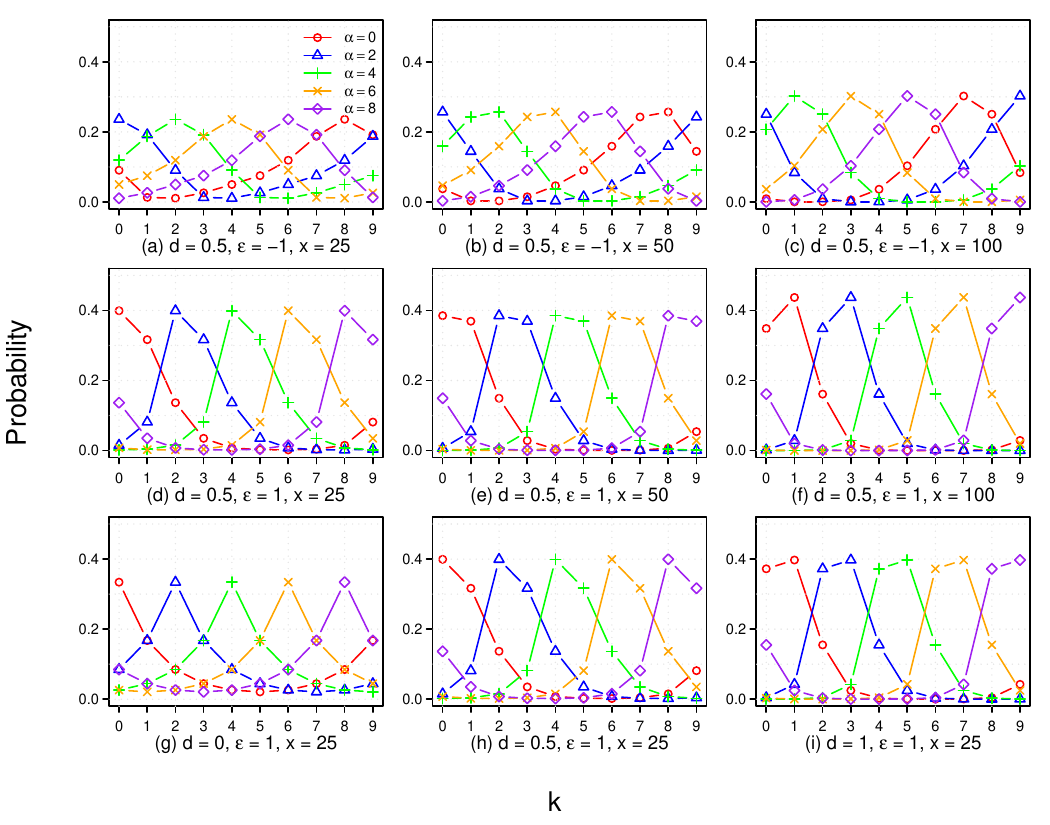}
   \caption{Plots of Conditional pmf of $\Theta$ given $X$}
    \label{fig:3.5}
\end{figure}
 
 In \autoref{fig:3.5}, the integers $k$ on the horizontal axis represent the discrete circular categories, where each $k$ corresponds to the angle $2 \pi k/\mathrm{m}$. In subfigures (a)--(c), the location of the mode shifts from $2\pi(\alpha+8)/\mathrm{m}$ to $2\pi(\alpha+7)/\mathrm{m}$ as $X$ increases from $25$ to $100$, where $\alpha+7$ and $\alpha+8$ are the neighboring categories. Similarly, in subfigures (d)--(f), this location moves from $2\pi\alpha/\mathrm{m}$ to $2\pi(\alpha+1)/\mathrm{m}$ over the same range of $X$. Subfigure (g) represents the marginal pmf of $\Theta$, where the variables are independent because $d = 0$. Here, the location of the mode is $2 \pi \alpha / \mathrm{m}$ (see subfigure (g)), but as $d$ increases, it shifts from $2\pi \alpha/\mathrm{m}$ to $2\pi(\alpha+1)/\mathrm{m}$. These variations show that the modal direction is influenced by $X$, $\varepsilon$, $d$, and $\alpha$, illustrating how the model captures circular--linear dependence through systematic shifts in the modal location. Note that the modal shift is not restricted to neighboring categories. For instance, under the same parameter settings as in subfigures (a)--(c), the mode shifts from $2\pi(\alpha+8)/\mathrm{m}$ to $2\pi(\alpha+6)/\mathrm{m}$ as X increases significantly (e.g., around $1200$).

In \autoref{fig:3.5}, the conditional distribution of $\Theta$ given $X = x$ appears to be unimodal. However, the model is flexible, and bimodality can arise under certain parameter settings. Specifically, for $\mu = \pi$, the distribution is symmetric around $\alpha$. Furthermore, if $\mu = \pi$, $\varepsilon = 1$, and $\mathrm{m} \ge 4$ is an even integer, there exists a range of values for $x$ where the distribution has two modes: one at $\alpha$ and another at $(\alpha + \mathrm{m}/2) \bmod \mathrm{m}$. Proofs of these properties are provided in the Appendix.

 The distribution in \eqref{eq:3.7}, similar to the marginal distribution in \eqref{eq:3.1}, exhibits the same boundary behavior at $q=0$, where it degenerates to a point mass at $\theta = \tfrac{2\pi\alpha}{\mathrm{m}}$.  
For $q=1$, it simplifies to 

\begin{align}\label{eq:3.8}
 P(\theta \mid x) = \frac{1}{p_5} \exp\Big(-(\beta x)^{\gamma}\Big(1 - \tanh(d) \cos\Big(\mu-\frac{2 \pi \zeta}{\mathrm{m}} \Big)\Big)^\varepsilon\Big), 
\end{align}
where $p_5 = \sum_{k = 0}^{\mathrm{m}-1}\exp \big(-(\beta x)^{\gamma}\big(1 - \tanh(d) \cos (\mu-\frac{2 \pi \zeta}{\mathrm{m}} )\big)^\varepsilon \big)$. Furthermore, when $d=0$, \eqref{eq:3.8} reduces to the uniform distribution.

\subsection{Conditional Distribution of X given \texorpdfstring{$\Theta$}{Theta}}
The conditional distribution of $X$ given $\Theta = \theta$ is 
\begin{align}\label{eq:3.11}
 f_{X|\theta}(x \mid \theta) =  &\gamma \beta^{\gamma} \Big(1 - \tanh(d) \cos\Big(\mu-\frac{2\pi \zeta}{\mathrm{m}}\Big)\Big)^\varepsilon x^{\gamma-1}\exp\Big(-(\beta x)^ \gamma \Big(1  \nonumber \\ &  - \tanh(d) \cos\Big(\mu-\frac{2\pi \zeta}{\mathrm{m}}\Big)\Big)^\varepsilon \Big). 
 \end{align}

\autoref{fig:3.9} presents the density of the conditional distribution given in \eqref{eq:3.11} for various values of $\beta,~\gamma$, and $d$ under the settings $\mathrm{m} = 10$, $q = 0.5$, $\mu = \frac{\pi}{2}$, $\alpha = 0$, and $\varepsilon = 1$, given that $\theta = 6 \pi/10$. In the first column, $\beta$ is fixed at $0.01$, while in the second column, $\gamma$ is fixed at $4$.

\begin{figure}[!ht]
 \centering
 \includegraphics[width= 1\linewidth]{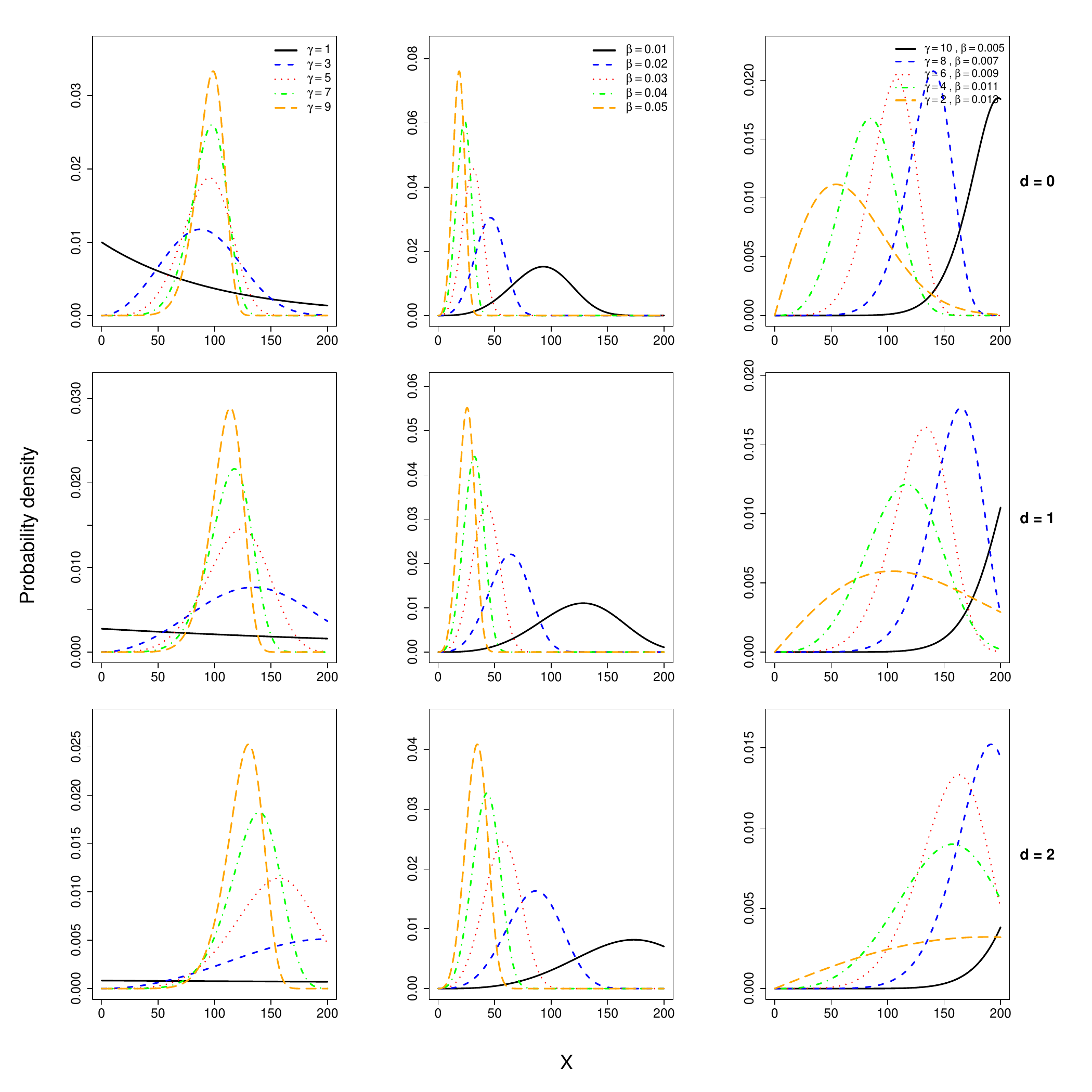}
 \caption{Density plots of Conditional distribution of $X$ given $\Theta$ }
 \label{fig:3.9}
\end{figure}

In \autoref{fig:3.9}, the subfigures in the first column show that larger values of $\gamma$ result in a highly concentrated distribution, whereas smaller values lead to a more dispersed form. The subfigures corresponding to the second column indicate that increasing $\beta$ shifts the distribution towards smaller values of $X$ and increases the density at the mode. The subfigures in the third column illustrate the joint influence of $\gamma$ and $\beta$. Overall, $\gamma$ governs the concentration of the conditional density, whereas $\beta$ influences both its location and scale. The parameter $d$ controls the interaction strength between the angular and linear components. The first row corresponds to $d = 0$, the second to $d = 1$, and the third to $d = 2$. \autoref{fig:3.9} shows that increasing $d$ results in a lower peak and a more spread-out distribution. We derive the interpretation of these parameters in the next section \ref{sec4}. 

Note that the conditional distribution given in \eqref{eq:3.11} depends on the parameters $\mu$ and $\alpha$ only through ($\mu - \frac{2\pi\zeta}{\mathrm{m}}$), where $\zeta = (k - \alpha) \bmod \mathrm{m}$. Consequently, the parameters are not individually identifiable. In the regression context, to ensure identifiability, we fix $\alpha = 0$ without loss of generality.

\section{Circular--Linear Regression}\label{sec4}
In this section, we formulate a regression model to analyze the dependence of the linear variable $X$ on the circular predictor $\Theta$. Since the conditional distribution of $X$ given $\Theta = \theta$ has a tractable closed form, it leads to a natural circular--linear regression. 
The conditional mean $\operatorname{E}(X \mid \Theta=\theta)$ and variance $\operatorname{V}(X \mid \Theta = \theta)$ are given as

\begin{align}\label{eq:3.12}
\operatorname{E}(X \mid \Theta = \theta) = \frac{\Gamma(1 + \frac{1}{\gamma})}{\beta \big(1 - \tanh(d) \cos(\mu-\frac{2 \pi \zeta}{\mathrm{m}})\big)^\frac{\varepsilon}{\gamma}},
 \end{align}

and 
\begin{align}\label{eq:3.13}
\operatorname{V}(X \mid \Theta = \theta) = \frac{\Gamma(1 + \frac{2}{\gamma}) - \Gamma(1 + \frac{1}{\gamma})^2}{\beta^2 \big(1 - \tanh(d) \cos(\mu - \frac{2 \pi \zeta}{\mathrm{m}})\big)^\frac{2 \varepsilon}{\gamma}}.
\end{align}

\subsection{Interpretation of Parameters}\label{sec4.1}
The parameters of the proposed model possess distinct interpretable properties. The parameter $\beta$ acts as a scaling factor, determining the baseline magnitude of the linear variable $X$. As seen in equations \eqref{eq:3.12} and \eqref{eq:3.13}, both the conditional mean and variance are strictly decreasing functions of $\beta$.  

With fixed $\alpha = 0$, the parameter $\mu$ provides a continuous adjustment to determine the angular position of the peak.

The shape parameter $\gamma$ governs the smoothness of the regression curve. We establish the following theoretical monotonicity property.

\begin{theorem}\label{theo:4.1}
If $\mu - \tfrac{2\pi\zeta}{\mathrm{m}}$ lies in the first or fourth quadrant with $\varepsilon = 1$,  
or in the second or third quadrant with $\varepsilon = -1$,  
then the conditional mean and variance of $X$ given $\Theta = \theta$ are decreasing functions of $\gamma$.
\end{theorem}

Similarly, when $\mu - \tfrac{2\pi\zeta}{\mathrm{m}}$ lies in the first or fourth quadrant, the conditional mean and variance are increasing functions of $d$ for $\varepsilon = 1$, and decreasing functions for $\varepsilon = -1$. Conversely, when the angle lies in the second or third quadrant, these relationships are reversed.   
Crucially, the parameter $d$ specifically quantifies the strength of dependence between the circular variable $\Theta$ and linear variable $X$.
When $d=0$, the term $\tanh(d)$ vanishes, and $\operatorname{E}[X \mid \Theta=\theta]$ becomes constant, indicating independence between $X$ and $\Theta$. Consequently, testing the hypothesis $H_0\!: d=0$ offers a robust framework for assessing the statistical significance of circular--linear association.

\section{Interdependence of Linear and Circular Variables}\label{sec5}
Linear and circular variables exhibit fundamentally different characteristics: linear variables take values along the real line with a natural ordering and infinite range, whereas circular variables are defined on the unit circle, exhibiting periodic behavior with no fixed origin. To analyze the association between such variables, specialized techniques that account for the directional nature of circular data are required. The traditional Pearson product-moment correlation \cite{pearson1901} is inappropriate in this context, as it is defined for measuring the association between linear variables only.  

To address this, K. V. Mardia \cite{mardia1976} introduced a correlation coefficient designed specifically to quantify the dependence between a circular and a linear variable. This measure embeds the angular variable through its sine and cosine transformations, 
leading to a squared correlation coefficient denoted by $\rho_{x\theta}^2$, defined as:

\begin{equation}\label{eq:4.1}
\rho^2_{x\theta} = \frac{r_{xc}^2 + r_{xs}^2 - 2r_{xc}r_{xs}r_{cs}}{1 - r_{cs}^2}.
\end{equation}
The terms  
$r_{xc} = \operatorname{corr}\!\left(X, \cos\Theta\right)$,  
$r_{xs} = \operatorname{corr}\!\left(X, \sin\Theta\right)$, and  
$r_{cs} = \operatorname{corr}\!\left(\cos\Theta, \sin\Theta\right)$  
denote the corresponding Pearson correlation coefficients.  
To compute $\rho^2_{x\theta}$, we use the following moment expressions.

\begin{align*}
&\operatorname{E}(X\cos \Theta) = C_2 \frac{\Gamma(1 + \frac{1}{\gamma})}{\gamma \beta^{1 + \gamma}} \sum_{k = 0}^{\mathrm{m}-1} \cos \tfrac{2 \pi k}{\mathrm{m}} (q^\zeta + q^{\mathrm{m}- \zeta})\big/H^{\varepsilon (1 + \frac{1}{\gamma})},\\
&\operatorname{E}(X\sin \Theta) =  C_2 \frac{\Gamma(1 + \frac{1}{\gamma})}{\gamma \beta^{1 + \gamma}} \sum_{k = 0}^{\mathrm{m}-1} \sin \tfrac{2 \pi k}{\mathrm{m}}(q^\zeta + q^{\mathrm{m}- \zeta})\big/H^{\varepsilon (1 + \frac{1}{\gamma})},\\
&\operatorname{E}(X) = C_2 \frac{\Gamma(1 + \frac{1}{\gamma})}{\gamma \beta^{1 + \gamma}} \sum_{k = 0}^{\mathrm{m}-1} (q^\zeta + q^{\mathrm{m}- \zeta})\big/H^{\varepsilon (1 + \frac{1}{\gamma})},\\
&\operatorname{E}(X^2) = C_2 \frac{\Gamma(1 + \frac{2}{\gamma})}{\gamma \beta^{2 + \gamma}} \sum_{k = 0}^{\mathrm{m}-1} (q^\zeta + q^{\mathrm{m}- \zeta})\big/H^{\varepsilon (1 + \frac{2}{\gamma})},\\
&\operatorname{E}(\cos^i \Theta \sin^j\Theta) = \frac{C_2}{\gamma \beta^{\gamma}}\sum_{k = 0}^{\mathrm{m}-1} \cos^i \tfrac{2 \pi k}{\mathrm{m}} \sin^j \tfrac{2 \pi k}{\mathrm{m}} (q^\zeta + q^{\mathrm{m}- \zeta}) \big / H^{\varepsilon},
\end{align*}
where $ H = 1 - \tanh{d} \cos(\mu - \frac{2 \pi \zeta}{\mathrm{m}})$ and $i,j \in \{0,1,2\}.$

\begin{theorem}\label{theo.5}
For the BWBWG distribution, the circular--linear correlation coefficient $\rho^2_{x\theta}$ equals zero whenever $d=0$.
\end{theorem}

The result of \autoref{theo.5} highlights the importance of the parameter $d$ in the dependency structure between $\Theta$ and $X$. \autoref{fig:4.1} shows contour plots of the correlation $\rho^2_{x\theta}$ as a function of $(d,\gamma)$ for fixed values $\mathrm{m} = 10$, $\alpha = 2$, $\beta = 0.5$, $\varepsilon = 1$, and $\mu = \pi/1.62$.

Subplot (a), corresponding to $q=0.1$, exhibits strong contour variation, reaching high correlation values ($\geq 0.85$) that change rapidly. As $q$ increases in subplots (b) and (c), the contour lines spread out and the maximum values drop significantly (to approximately 0.35). This indicates that increasing $q$ diminishes the correlation, making it significantly less sensitive to parameter changes. The parameter $d$ exhibits a positive effect on the correlation. While the magnitude of this correlation is modulated by $\gamma$, increasing $d$ consistently raises $\rho^2_{x\theta}$, confirming its role in the dependence between the variables.  

\begin{figure}[!ht] 
 \centering    
    \centering
    \includegraphics[width= 1\textwidth, height=4.5cm]{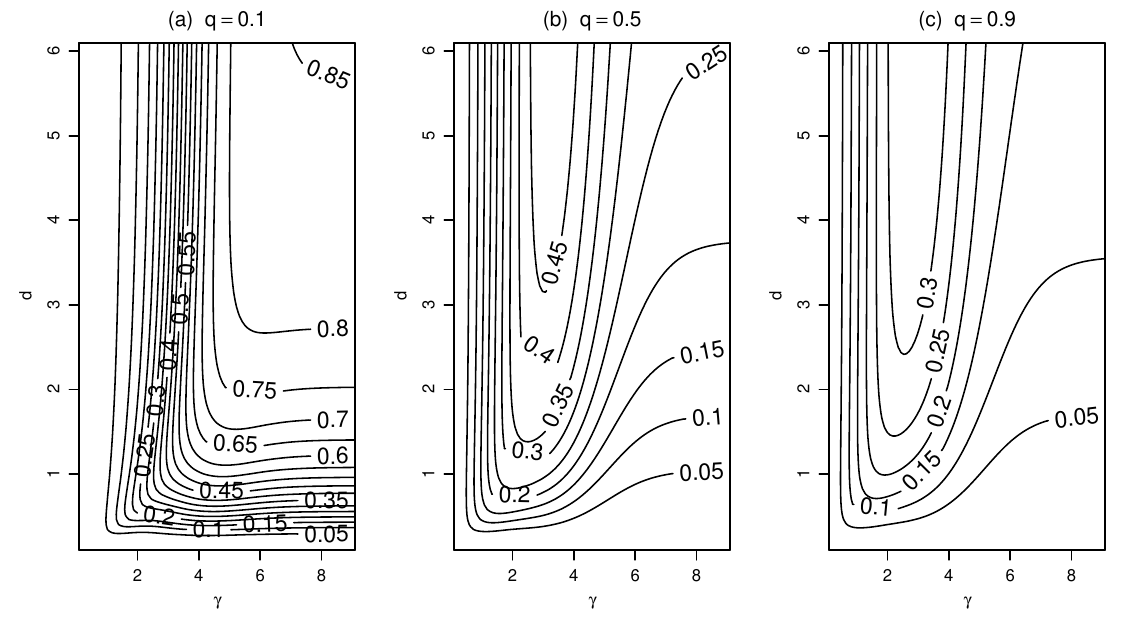} 
    \caption{Contour plot of circular--linear correlation $ \rho^2_{x\theta} $}
    \label{fig:4.1}
\end{figure}


\section{Simulation Study}\label{sec6}
In this section, the finite-sample properties of the maximum likelihood estimators for the BWBWG distribution are investigated via a simulation study. We consider varying sample sizes $n \in \{50, 200, 500\}$ across three distinct parameter configurations for $\varepsilon, \alpha, q, d, \mu, \beta$, and $\gamma$. For each case, $1000$ independent Monte Carlo replicates are generated. The random number generation algorithm, the estimation procedure used in each replicate, and a detailed discussion of the results are presented below.

\subsection{Random Number Generation}
 To generate random samples from the BWBWG distribution of $(\Theta, X)$, we employ a two-stage procedure based on the factorization of $f(\theta,x)$ into $P(\Theta = \theta)f(x \mid \Theta = \theta)$.
\begin{itemize}
    \item \textbf{Stage I:} The discrete circular variate $ \Theta = 2\pi K/\mathrm{m}$ is generated from its marginal pmf given in \eqref{eq:3.1}, using the inverse transform sampling method. Given the cumulative probabilities $F_k = \sum_{j=0}^{k} P(\Theta = 2\pi j/\mathrm{m}) $, we generate a uniform random number $u \sim U(0,1)$ and set $K = k$ such that
$
    F_{k-1} < u \le F_{k}
$ 

\item \textbf{Stage II:} Conditioned on the realized value $2\pi k/\mathrm{m}$ obtained in the first stage, the continuous linear variate $X$ is generated from the conditional density function given in \eqref{eq:3.11}, by using the inverse transform sampling
method. We generate a uniform random number $v \sim U(0,1)$ and compute $X$ using the closed-form inversion: 
\begin{equation*}
 X = \frac{1}{\beta} \left(\frac{-\log(1 - v)}{\big(1 - \tanh (d) \cos (\mu - \frac{2 \pi \zeta}{\mathrm{m}})\big)^{\varepsilon}}\right)^{\frac{1}{\gamma}}   
\end{equation*}

\end{itemize}

\subsection{Estimation Procedure}
In the BWBWG distribution, the parameters $\varepsilon$ and $\alpha$ are discrete, whereas $q, d, \mu, \beta,$ and $\gamma$ are continuous. Therefore, the MLE procedure is performed in two steps.

\begin{itemize}
   \item \textbf{Step I:} Each of the possible combinations $(\varepsilon, \alpha) \in \{-1, 1\} \times \mathbf{Z}_\mathrm{m}$ is substituted into the log-likelihood function. This yields $2\mathrm{m}$ distinct objective functions, each involving the five continuous parameters.

    \item \textbf{Step II:} Each of these functions is numerically maximized with respect to the continuous parameters. The final estimates are obtained by selecting the combination of discrete and continuous parameters that yields the global maximum log-likelihood value across all evaluations.
\end{itemize}

The numerical maximization in Step II is carried out using the `optim' function (method `L-BFGS-B') in the R programming environment.

\subsection{Results and Discussions}
\autoref{tab:4.1} presents the simulation results, with standard errors provided in parentheses. The reported standard errors represent the sample standard deviation of the $1,000$ parameter estimates obtained from the simulations. Additionally, the frequencies of the discrete parameter estimates across all simulations are summarized in the last two columns of this table.

\begin{table}[!ht]
\caption{Estimates (with standard deviation in parentheses) for BWBWG }\vspace{-2mm}
\label{tab:4.1}
\centering
\begin{tabular}{cccccccc}
\toprule
 \multicolumn{1}{c}{Sample size} & \multicolumn{7}{c}{MLE and its standard deviation} \\
 \cmidrule(lr){1-8}
 n & $\hat{d}_{MLE}$ & $\hat{\beta}_{MLE}$ & $\hat{\gamma}_{MLE}$ & $\hat{\mu}_{MLE}$ & $\hat{q}_{MLE}$  & $\hat{\alpha}_{MLE}$ & $\hat{\varepsilon}$\\
 \hline & $d = 0.2$ & $\beta = 0.8$ & $\gamma = 1$ & $\mu = 3.5$ & $q = 0.2$ & $\alpha = 2$ & $\varepsilon = 1$\\
 \hline
  50 &  2.692 & 1.262 & 1.052 & 6.280  & 0.223 & (0, 1, 997, 2,\dots,0) & (537, 463)\\
     & (3.212) & (0.880) & (0.122) & (2.060) & (0.070) & & \\
 200 & 1.363 & 0.941 & 1.010 & 2.584 & 0.208 & (0, 0, 1000,\dots,0) & (554, 446)\\
     & (2.433) & (0.288) & (0.055) & (2.422) & (0.031) & & \\
 500 & 0.508 & 0.866 & 1.005 & 3.866 & 0.204 & (0, 0, 1000,\dots,0) & (533, 467) \\
     & (1.100) & (0.187) & (0.035) & (2.594) & (0.020) &   &  \\
 \hline
 True value & $d=0.5$ & $\beta = 0.5$ & $\gamma = 1$ & $\mu = 3.5$ & $q=0.5$ &  $\alpha = 2$ & $\varepsilon = 1$\\
 \hline
  50 & 0.573 & 0.509 & 1.063 & 3.555 & 0.496 & (1, 41, 901, 57,\dots,0) & (428, 572) \\
     & (0.235) & (0.108) & (0.125) & (1.987) & (0.066) & &  \\
 200 &  0.506 & 0.493 & 1.014 & 3.529 & 0.500 & (0, 0, 1000,\dots,0)  & (251, 749)\\
     & (0.109) & (0.046) & (0.058) & (1.195) & (0.031) & & \\
 500 & 0.500  &  0.496 & 1.005 & 3.511 & 0.499 & (0, 0, 1000,\dots,0) & (144, 856) \\
     & (0.074) & ( 0.031) & (0.035) & (0.832) & (0.020) &  & \\
 \hline
 True value & $d= 1$ & $\beta = 0.2$ & $\gamma = 1.5$ & $\mu = 1.5$ & $q=0.8$ &  $\alpha = 5$ & $\varepsilon = -1$ \\ 
\hline
 50  & 1.106 & 0.200 & 1.521 & 1.375 & 0.780 & (494, 191, 87, 15, 6, & (991, 9)\\
 & & & & & & 9, 9, 12, 37, 140) &  \\
   & (0.395) & (0.011) & (0.086) & (0.742) & (0.049) & & \\ 
 200 & 1.086  & 0.201 & 1.521 & 1.409  & 0.783 & (466, 193, 68, 28, 2, & (988, 12) \\
 & & & & & & 9, 11, 13,  45, 165) &  \\
      & (0.394) & (0.011) & (0.084) & (0.781) & (0.047) & & \\
 500 & 1.025 & 0.200 & 1.505  & 1.465 & 0.798 & (730, 141, 14, 2, 1, & (1000, 0) \\minai
 
 & & & & & & 0, 0, 0, 8, 104) & \\
    & (0.105) & (0.007) & (0.050) & (0.377) & (0.035) & & \\
\hline
\end{tabular}
\end{table}

The estimates of the discrete parameters $\alpha$ and $\varepsilon$ exhibit distinct behaviors with respect to the parameters $q$ and $d$, respectively. For small values of $q$ (e.g., $q = 0.2$), the estimates of $\alpha$ tend to be more accurate than those obtained for larger values of $q$ (e.g., $q = 0.8$). To study this, the distributions for $q \in \{0,1\}$ are derived by evaluating the limits of the joint density given in \eqref{eq:2.5}, and the behaviour of the estimates is expected. This is because when $q=0$, the BWBWG distribution in \eqref{eq:2.5} collapses to the following form

\begin{align}\label{eq:5.1}
\resizebox{1\textwidth}{!}{$ 
f(\theta,x) =
\begin{cases}
\gamma \beta^\gamma \Big(1 - \tanh(d)\cos(\mu)\Big)^\varepsilon x^{\gamma-1} \exp\Big(-(\beta x)^{\gamma}\big(1 - \tanh(d)\cos (\mu)\big)^\varepsilon\Big), & \theta = \frac{2 \pi \alpha}{\mathrm{m}},\quad x \in (0, \infty) \\
0, \quad \text{otherwise}.
\end{cases}
$}
\end{align}

This is a degenerate distribution supported entirely on the ray $\{2\pi\alpha/\mathrm{m}\} \times (0, \infty)$, which makes the parameter $\alpha$ highly identifiable.

In contrast, when $q = 1$, the distribution flattens with respect to $\theta$. In this case, the distribution simplifies to  
\begin{equation*}
\resizebox{0.85\textwidth}{!}{$
f(\theta,x) = \frac{\gamma \beta^\gamma}{p_3} x^{\gamma-1} \exp\bigg(-(\beta x)^{\gamma}\Big(1 - \tanh(d)\cos\big(\mu - \frac{ 2 \pi \zeta}{\mathrm{m}}\big)\Big)^\varepsilon\bigg).
$}
\end{equation*}

The parameter $q$ also influences estimates of the continuous parameters. As seen in the degenerate case given in \eqref{eq:5.1}, for $ q= 0$, the joint density is supported entirely on the ray $\{2\pi\alpha/\mathrm{m}\} \times (0, \infty)$, which removes the directional variability needed to estimate the orientation and dependence structure. Consequently, as $q$ increases, this variability is restored, leading to more precise estimates. The effect of the parameter $q$ can be seen in \autoref{tab:4.1} (specifically, note the decreasing standard errors of $\mu$ and $d$).

Similarly, the estimates of $\varepsilon$ are more accurate for larger values of $d$ (e.g., $d = 1$), and become less reliable when $d$ is small (e.g., $d = 0.2$). This behavior occurs because the role of $\varepsilon$ in the BWBWG distribution is governed by the term
\begin{equation*}
\left(1 - \tanh(d)\cos\big(\mu - \tfrac{2\pi\zeta}{\mathrm{m}}\big)\right)^{\varepsilon},
\end{equation*}
which is highly dependent on the parameter $d$. Specifically, if $d = 0$, then the influence of this term on the distribution vanishes. Consequently, for small values of $d$, the parameter $\varepsilon$ becomes difficult to identify, leading to reduced estimation accuracy. 


\section{Data Analysis}\label{sec7}
This section presents the practical applicability of the proposed methodology through the analysis of two distinct datasets: a wind dataset from Gwalior, India, and an SO$_2$ concentration-acrophase dataset from Delhi, India. We fit the two distinct BWBWG models, corresponding to the discrete parameter values $\varepsilon = -1$ and $\varepsilon = 1$, to both datasets. Additionally, for the SO$_2$ concentration-acrophase dataset, we perform the regression analyses for both values of $\varepsilon$. Finally, we compare the results with the discretized versions of some well-known models.

\subsection{Application to Wind Dataset}
The wind dataset consists of measurements of wind direction and speed recorded in Gwalior, India, between March~22 and May~4, 2024. Observations were taken approximately every three hours, and time points with no wind were excluded from the dataset. The wind‐direction measurements are categorical and are encoded using $16$ integers, each corresponding to one of the $16$ standard compass sectors; the mapping is shown in \autoref{fig:5.1}. 
\begin{figure}[!ht]
    \centering
    \includegraphics[width=0.8\linewidth]{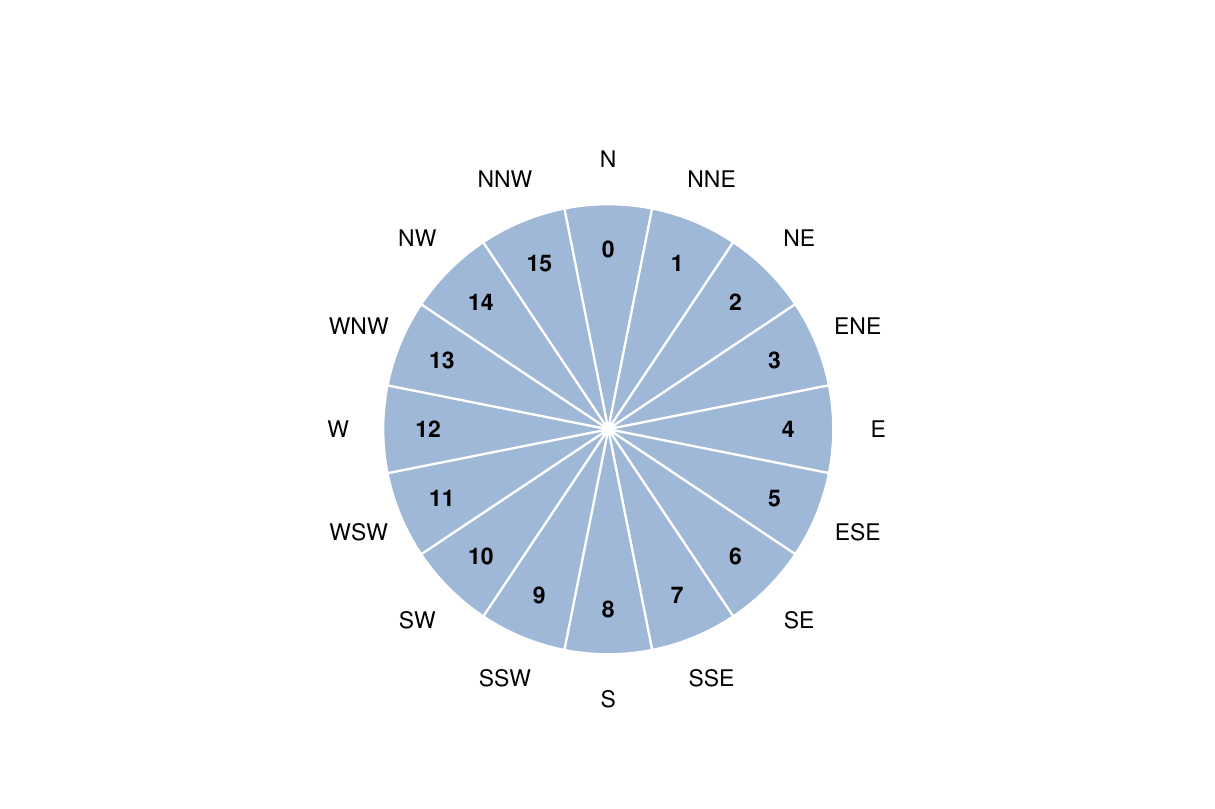}\vspace{-5 mm}
    \caption{Mapping of wind directions to their integer representations}
    \label{fig:5.1}
\end{figure}

\subsubsection{BWBWG Models Fitting}
The BWBWG model with $\varepsilon = 1$ has a smaller Akaike Information Criterion (AIC) value than the model with $\varepsilon = -1$. Therefore, we consider the BWBWG model (with $\varepsilon = 1$) for further analysis. The estimates of the parameters are reported in \autoref{tab:5.1}, and the corresponding heatmaps of the estimated density are displayed in \autoref{fig:5.2}.

\begin{table}[!ht]
\centering
\small
\caption{Estimates of the BWBWG  parameters ($\varepsilon = 1$) for the wind dataset.}\vspace{-3mm} 
\label{tab:5.1}
\begin{tabular}{m{4cm} | m{6cm} }  
\hline
 $\textbf{Parameters}$ & \textbf{Estimates (Sd Error)} \\
 \hline
 $\hat{q}_{MLE}$ & 0.873 (0.034)  \\
 $\hat{d}_{MLE}$ & 0.823 (0.073) \\
 $\hat{\mu}_{MLE}$ & 2.392 (0.073) \\
 $\hat{\beta}_{MLE}$ & 0.298 (0.013) \\
 $\hat{\gamma}_{MLE}$ & 1.781 (0.080) \\ 
 $\hat{\alpha}_{MLE}$ & 7 \\ 
 \hline
\end{tabular}
\end{table}

\begin{figure}[!ht]
 \centering
 \includegraphics[width=0.9\linewidth, height=0.3\textheight]{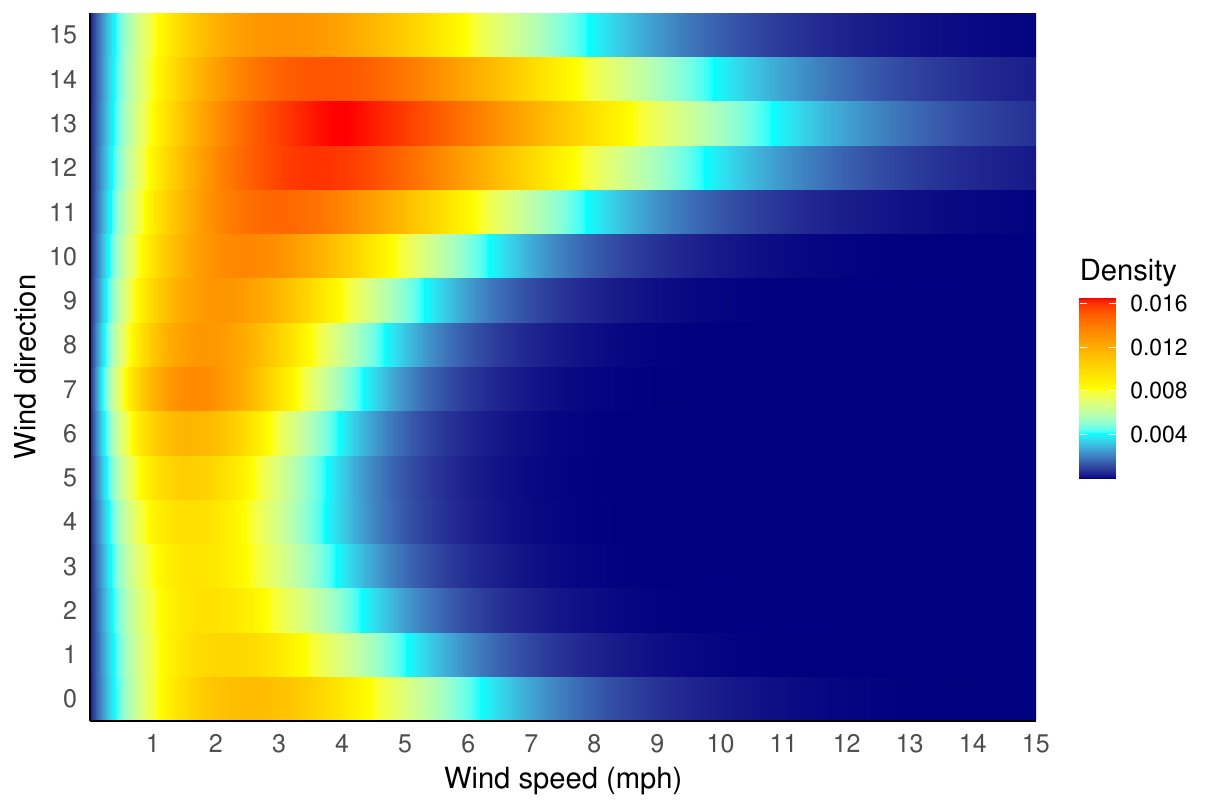}
\caption{Heatmap of the estimated BWBWG density ($\varepsilon = 1$).}
\label{fig:5.2}
\end{figure}

\textit{Interpretation of Estimates:} The estimate $\hat{d}_{MLE} = 0.823$ indicates dependence between wind direction and wind speed. As shown in the heatmap in \autoref{fig:5.2}, the highest density corresponds to westerly winds with low speeds, typically between 2 and 4 miles per hour. A secondary, lower-density region associated with easterly winds is also visible, while winds blowing directly from the north or south occur much less frequently.

For comparison purposes, we also apply the discretized circular formulation of existing models to the wind dataset. Since no model of the same mixed type (discrete circular and continuous linear) is available in the literature, the competing models are adapted by numerically integrating over the circular variable within the corresponding intervals. The resulting negative log-likelihood and AIC values are summarized in \autoref{tab:5.2}.

\begin{table}[!ht]
\centering
\caption{Comparison of the BWBWG model with discretized versions of existing models: Weibull sine-skewed von Mises (WeiSSVM), Generalized Gamma sine-skewed von Mises (GGSSVM), and Johnson--Wehrly (JW).}
\begin{tabular}{m{4cm} | m{4cm} | m{2.5cm} }
\hline
\textbf{Model} & \textbf{Negative Log-Likelihood} & \textbf{AIC} \\
\hline
BWBWG ($\varepsilon = -1$)  & 1309.537 & 2631.074 \\
BWBWG ($\varepsilon = 1$)  & 1295.664 & 2603.328 \\
Discretized WeiSSVM  & 1297.514 & 2605.028 \\
Discretized GGSSVM   & 1297.166 & 2606.332 \\
Discretized JW   &  1356.907 &  2719.814 \\
\hline
\end{tabular}
\label{tab:5.2}
\end{table}

The results in \autoref{tab:5.2} show that the BWBWG model ($\varepsilon = 1$) yields the lowest negative log-likelihood and AIC values. Therefore, our direct mixed-type (discrete circular and continuous linear) formulation outperforms the discretized versions of these existing models for the wind dataset.




\subsection{Application to \texorpdfstring{SO$_2$} Concentration-Acrophase Dataset}
The SO$_2$ Concentration-Acrophase Dataset (SO$_2$ dataset) is derived from high-frequency air quality monitoring records in Delhi, focusing specifically on the diurnal behavior of Sulfur Dioxide (SO$_2$) concentrations. The raw data is sourced from Kaggle, and the source link is `` \url{https://www.kaggle.com/datasets/arnavtripathi01/air-quality-and-metrological-data-in-india-2024}''. This dataset consists of hourly measurements spanning from January 1, 2024, to December 31, 2024. For each 24-hour period, we extracted two key variables:
\begin{enumerate}
    \item \textbf{Linear Variable ($X$):} The peak intensity of the daily cycle, defined as the maximum observed SO$_2$ concentration for that day. This variable represents the highest pollution level recorded within the 24-hour window.   
    \item \textbf{Circular Variable ($\Theta$):} The acrophase, representing the time of peak intensity. It is calculated as $\Theta = 2 \pi k/24$, where $k$ is the specific hour of the day at which the maximum concentration $X$ was recorded.
\end{enumerate}

The processed data consists of 366 observations and the descriptive Statistics is presented in \autoref{tab:5:3}, where mean values are reported in $\mu g/m^3$. Frequency denotes the count of daily SO$_2$ maxima.  
\begin{table}[!ht]
\centering
\caption{Hourly frequency and mean intensity of daily SO$_2$ peaks.}
\label{tab:5:3}
\begin{tabular}{l|cccccccccccc}
\hline
\multicolumn{13}{c}{\textbf{Morning to Noon (00:00 -- 11:00)}} \\
\hline
\textbf{Hour} & \textbf{0} & \textbf{1} & \textbf{2} & \textbf{3} & \textbf{4} & \textbf{5} & \textbf{6} & \textbf{7} & \textbf{8} & \textbf{9} & \textbf{10} & \textbf{11} \\ 
\textbf{Freq.} & 43 & 10 & 16 & 12 & 1 & 6 & 4 & 16 & 19 & 23 & 19 & 26 \\ 
\textbf{Mean} & 19.8 & 20.9 & 17.0 & 14.5 & 20.1 & 20.7 & 35.4 & 35.1 & 30.6 & 23.3 & 20.6 & 18.5 \\ 
\hline
\multicolumn{13}{c}{\vspace{-2mm}} \\ 
\multicolumn{13}{c}{\textbf{Afternoon to Night (12:00 -- 23:00)}} \\
\hline
\textbf{Hour} & \textbf{12} & \textbf{13} & \textbf{14} & \textbf{15} & \textbf{16} & \textbf{17} & \textbf{18} & \textbf{19} & \textbf{20} & \textbf{21} & \textbf{22} & \textbf{23} \\ 
\textbf{Freq.} & 20 & 12 & 15 & 6 & 5 & 8 & 1 & 11 & 10 & 11 & 35 & 37 \\ 
\textbf{Mean} & 17.4 & 16.2 & 16.3 & 18.5 & 14.9 & 14.4 & 14.2 & 15.1 & 15.4 & 17.9 & 19.7 & 19.4 \\ 
\hline
\multicolumn{13}{l}{\footnotesize}
\end{tabular}%
\end{table}

\subsubsection{BWBWG Models Fitting}
The BWBWG model with $\varepsilon = -1$ yields a smaller AIC value than the model with $\varepsilon = 1$. Therefore, we select the BWBWG model (with $\varepsilon = -1$) for further analysis. The estimated parameters are presented in \autoref{tab:5.6}, and the corresponding estimated joint density is displayed in \autoref{fig:5.6}. 

\begin{table}[!ht]
\centering
\small
\caption{Estimates of the BWBWG  parameters ($\varepsilon = -1$) for the SO$_2$ dataset.}\vspace{-3mm} 
\label{tab:5.6}
\begin{tabular}{m{4cm} | m{6cm} }  
\hline
 $\textbf{Parameters}$ & \textbf{Estimates (Sd Error)} \\
 \hline
 $\hat{q}_{MLE}$ & 0.777 (0.016)  \\
 $\hat{d}_{MLE}$ & 1.322 (0.104) \\
 $\hat{\mu}_{MLE}$ & 6.099 (0.040) \\
 $\hat{\beta}_{MLE}$ & 0.043 (0.001) \\
 $\hat{\gamma}_{MLE}$ & 2.832 (0.092) \\ 
 $\hat{\alpha}_{MLE}$ & 18 \\ 
 \hline
\end{tabular}
\end{table}

\begin{figure}[!ht]
 \centering
 \includegraphics[width=0.9\linewidth, height=0.3\textheight]{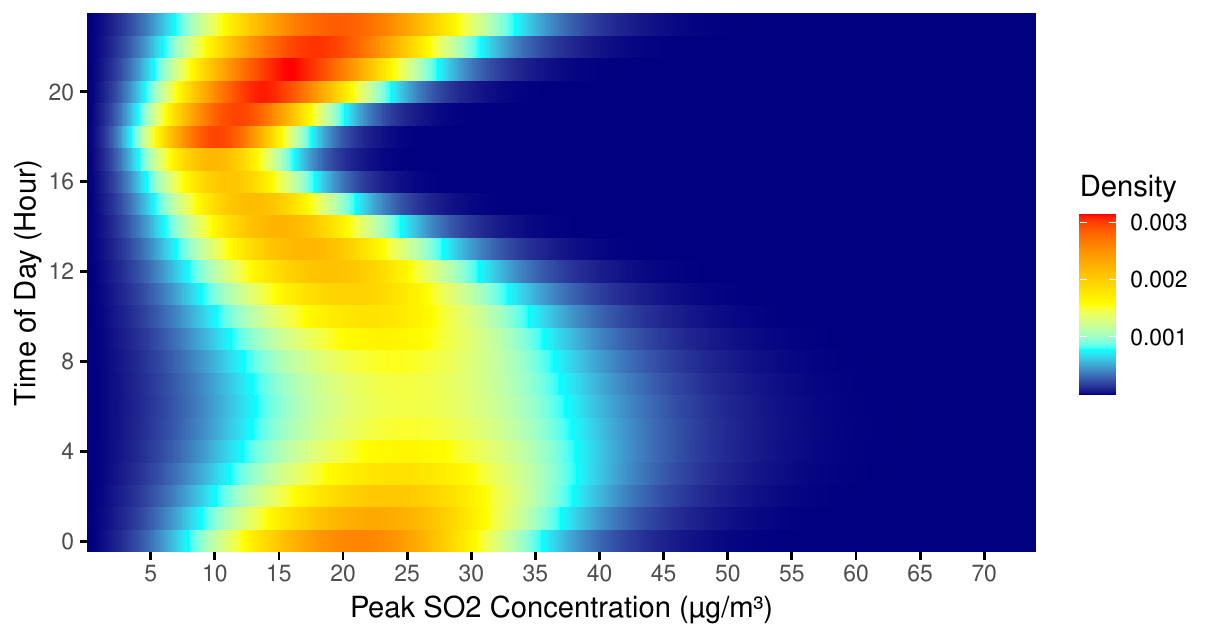}
\caption{Heatmap of the estimated BWBWG density ($\varepsilon = -1$).}
\label{fig:5.6}
\end{figure}

\textit{Interpretation of Estimates:} The estimate $\hat{d}_{MLE} = 1.322 $ confirms a dependence between the peak intensity of SO$_2$ concentration and the daily cycle. As illustrated in \autoref{fig:5.6}, higher concentrations are not uniformly distributed but cluster around specific hours, specifically peaking between 16:00 and 20:00 hours. 

For comparative analysis, we also fitted the discretized circular formulation of some existing continuous models; these results are summarized in \autoref{tab:5.7}. 

\begin{table}[!ht]
\centering
\caption{Comparison of the BWBWG model with discretized versions of existing models: Weibull sine-skewed von Mises (WeiSSVM), Generalized Gamma sine-skewed von Mises (GGSSVM), and Johnson--Wehrly (JW).}
\begin{tabular}{m{4cm} | m{4cm} | m{2.5cm}}
\hline
\textbf{Model} & \textbf{Negative Log-Likelihood} & \textbf{AIC} \\
\hline
BWBWG ($\varepsilon = -1$)  & 2381.861 & 4775.722 \\
BWBWG ($\varepsilon = 1$)  & 2401.380 & 4814.760 \\
Discretized WeiSSVM  & 2422.218  & 4854.436 \\
Discretized GGSSVM   & 2376.059  & 4764.118 \\
Discretized JW   &  2624.320 &  5254.641 \\
\hline
\end{tabular}

\label{tab:5.7}
\end{table}

The results in \autoref{tab:5.7} show that both variants of the BWBWG model ($\varepsilon = -1$ and $\varepsilon = 1$) significantly outperform the discretized WeiSSVM and JW models. While the discretized GGSSVM achieves the lowest statistics, our direct mixed-type (discrete circular and continuous linear) formulation yields the second-best performance for the SO$_2$ dataset.

\subsubsection{Regression Analysis}
The objective of the regression analysis is to describe how the expected linear response varies with the circular variable by fitting the regression model given in \eqref{eq:3.12}. Specifically, the conditional expectation represents the mean peak SO$_2$ concentration conditioned on the time of occurrence. 

The regression model is fitted to the SO$_2$ concentration dataset for both values of $\varepsilon$. The parameter estimates were obtained by MLE, and the model with $\varepsilon = 1$ yields a lower AIC value compared to the model with $\varepsilon = -1$. Therefore, we select the regression model ($\varepsilon = 1$) for further analysis. The resulting estimates are reported in \autoref{tab:5.4}, and the corresponding fitted regression curve is presented in \autoref{fig:5.5}.

\begin{table}[!ht]
\centering
\small
\caption{Estimated parameters of the regression model ($\varepsilon = 1$) for the SO$_2$ dataset}\vspace{-3mm} 
\label{tab:5.4}
\begin{tabular}{ m{4cm} | m{6cm}  }  
\hline
 $\textbf{Parameters}$ & \textbf{Estimates (Sd Error) } \\
 \hline
 $\hat{d}_{MLE}$ & 0.895 (0.114) \\ 
 $\hat{\mu}_{MLE}$ & 1.456 (0.075)  \\
 $\hat{\beta}_{MLE}$ & 0.047  (0.001) \\
 $\hat{\gamma}_{MLE}$ & 2.705 (0.089 ) \\  
 \hline
\end{tabular}
\end{table}

\begin{figure}[!ht]
    \centering
    \includegraphics[width= 0.9\linewidth]{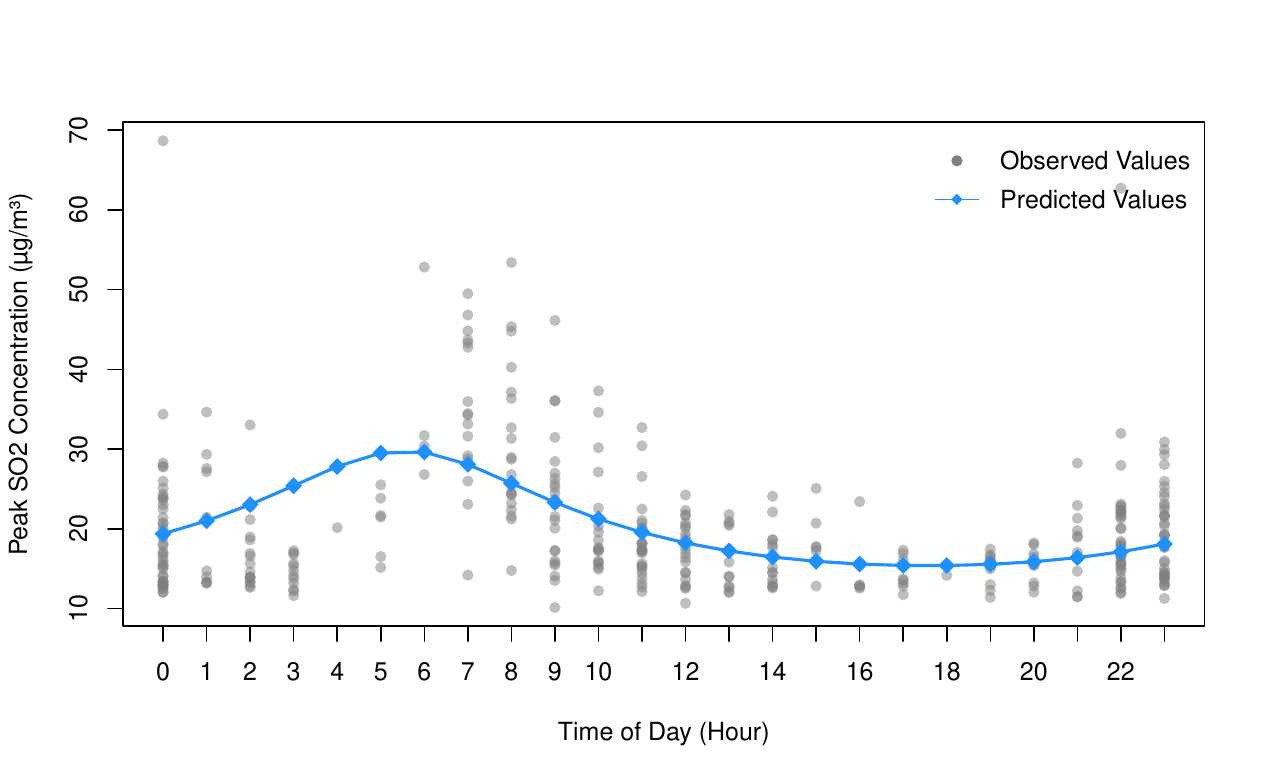}
    \caption{Diurnal variation of maximum daily SO$_2$ concentrations fitted with a circular--linear regression curve ($\varepsilon = 1$).}
    \label{fig:5.5}
\end{figure}

\textit{Testing for Independence:} As we have discussed in Section \ref{sec4.1}, the parameter $d$ quantifies the strength of the dependence between the circular variable $\Theta$ (time of day) and the linear variable $X$ (SO$_2$ concentration). Furthermore, when $d=0$, the conditional expectation becomes constant, indicating independence between $X$ and $\Theta$. To assess the statistical significance of this association, we tested the null hypothesis $H_0: d=0$ using the Likelihood Ratio Test. For the model ($\varepsilon = 1$), the test statistic is $ \Lambda = 2 \big( \text{NLL}_{\mathrm{H}_0} - \text{NLL} \big) $, where $\text{NLL}_{\mathrm{H}_0}$ and $\text{NLL}$ represent the negative log-likelihood values of the restricted model (under the null hypothesis) and the unrestricted model, respectively. The resulting test statistic is $\Lambda = 2(1278.313 - 1236.953) = 82.720$. Evaluated against a $\chi^2_1$ distribution, this yields a $p$-value of $< 0.001$. Consequently, we strongly reject the null hypothesis. 

\textit{Interpretation of Estimates:} The estimate $\hat{d}_{MLE} = 0.895$ confirms a circular--linear dependence, indicating that the magnitude of SO$_2$ peaks is conditional on their timing.

\autoref{fig:5.5} illustrates the diurnal behavior of SO$_2$ concentrations as captured by the discrete circular--linear regression model.

For comparison, we also fitted the discretized versions of the regression models proposed by \cite{johnson1978} and \cite{abe2017}; the results are presented in \autoref{tab:5.5}.

\begin{table}[!ht]
\centering
\caption{Comparison of the regression model with the discretized versions of the Johnson \& Wehrly (1978) and Abe \& Ley (2017) regression models for the SO$_2$ dataset.}
\begin{tabular}{m{4cm} | m{4cm} | m{2.5cm} }
\hline
\textbf{Model} & \textbf{Negative Log-Likelihood} & \textbf{AIC} \\
\hline
\textbf{Model ($\varepsilon = -1$)} &   1240.328 & 2488.656 \\
\textbf{Model ($\varepsilon = 1$) } &  1236.953 & 2481.906 \\
\textbf{Abe \& Ley} &  1236.953 & 2481.906 \\
\textbf{Johnson \& Wehrly} &  1460.103 &  2926.205 \\
\hline
\end{tabular}
\label{tab:5.5}
\end{table}

The results in \autoref{tab:5.5} demonstrate that the regression model for both values of $\varepsilon$ significantly outperforms the discretized Johnson \& Wehrly model. Notably, the  model ($\varepsilon = 1$) yields negative log-likelihood and AIC values identical to those of the discretized Abe \& Ley model for the SO$_2$ dataset.

In summary, we compared the proposed methods against discretized versions of well-known models. The results indicate that the proposed approach either outperforms or performs comparably to these alternatives. Furthermore, it offers a computationally efficient solution for datasets involving discrete circular and continuous linear components. Crucially, the proposed method avoids the approximation errors and computational burden inherent in the discretization process, thereby providing a robust framework for modeling circular--linear data.


\section{Conclusion}\label{sec8}

This article introduces a new methodology for handling challenges in cylindrical data, specifically those involving discrete circular and continuous linear variables. The proposed BWBWG distribution is studied in detail, with emphasis on parameter interpretations, dependency structures, and the associated marginal and conditional distributions. The performance of the proposed model is evaluated through applications to wind and SO2 Concentration-Acrophase dataset, where the results are compared against discretized versions of well-known models. Beyond these datasets, the proposed framework has wide applicability to other domains that combine directional and linear measurements, including ocean-current studies (direction and magnitude), animal movement analysis (direction and step length), and aviation datasets (flight heading and altitude or speed), among others.

In the case study, we analyzed a wind dataset from Gwalior, India, containing measurements of wind direction and speed using the proposed joint model. Additionally, for the SO$_2$ Concentration-Acrophase dataset from Delhi, India, we applied both the joint model and the circular--linear regression formulation to capture the diurnal behavior of Sulfur Dioxide (SO$_2$) concentrations.

Future studies may expand on the current statistical framework. Specifically, while the bivariate models introduced by \cite{dhakad2026} are defined for ordinal circular variables, the BWBWG model is defined for an ordinal circular and a continuous linear variable. However, there is currently no multivariate model for multiple ordinal circular and continuous linear variables. Furthermore, the proposed framework can also be extended to account for temporal dependence in time-series data.

\section*{Funding statement}
\vspace{-2mm}
 This project is partially funded by the Government of India as part of the Start-up Research Grant. The funding is provided through the Science of Engineering Board of the Department of Science and Technology to Dr. Jayant Jha (SRG/2022/000151). Brajesh Kumar Dhakad is funded by the National Board for Higher Mathematics (NBHM) (0203/4/2022/R\&D-II/ 3660). The funders had no role in the methodological development, data analysis, decision to publish, or preparation of the manuscript.
 \vspace{-0.5cm}

\section*{Disclosure statement}
\vspace{-0.2cm}
The authors report there are no competing interests to declare.
 \vspace{-0.4cm}
\section*{Data availability statement}
\vspace{-0.3cm}
The data supporting the findings of this study are available upon request.
\vspace{-0.4cm}


\begin{appendices} \label{app}

\section{Proof of \autoref{theo:2.5} }
\begin{proof}
Assume $d = 0$.  
Then $\tanh (d)=0$, and the joint density \eqref{eq:2.5} simplifies to
\begin{equation*}
f(\theta,x) = \frac{\gamma \beta^\gamma (1-q)}{(1-q^\mathrm{m})(1+q)}\Big(q^\zeta + q^{\mathrm{m}-\zeta}\Big)x^{\gamma-1} \exp\Big(-(\beta x)^{\gamma}\Big)    
\end{equation*}
 The marginal distribution of $\Theta$ is obtained by integrating the simplified joint density over the support of $X$, yielding
 \begin{equation*}
 P( \theta) = \frac{(1-q)}{(1-q^\mathrm{m})(1+q)} \Big(q^\zeta + q^{\mathrm{m}-\zeta}\Big)     
 \end{equation*}
The marginal density of $X$ is obtained by summing the simplified joint density over the support of $\Theta$, yielding
 \begin{equation*}
  f(x) = \gamma \beta^\gamma x^{\gamma-1} \exp\Big(-(\beta x)^{\gamma}\Big)   
 \end{equation*}
Since the joint density factorizes as
\[
f(\theta,x)
= P(\theta)\,f(x), 
\]
the random variables $\Theta$ and $X$ are independent.

Conversely, suppose that the random variables $\Theta$ and $X$ are independent. It follows that the conditional distribution of $X$ given $\Theta = \theta$,
\begin{align*}
    f_{X|\theta}(x \mid \theta) &= \gamma \beta^{\gamma} \Big(1 - \tanh(d) \cos\Big(\mu-\frac{2\pi \zeta}{\mathrm{m}}\Big)\Big)^\varepsilon x^{\gamma-1} \\&\quad \exp\left\{-(\beta x)^ \gamma \Big(1 - \tanh(d) \cos\Big(\mu-\frac{2\pi \zeta}{\mathrm{m}}\Big)\Big)^\varepsilon \right\},
\end{align*}
must be invariant with respect to $\theta$. This implies the coefficient $\tanh(d)$ of $\cos\Big(\mu-\frac{2\pi \zeta}{\mathrm{m}}\Big)$ must be 0; thus, $d = 0$.

\end{proof}


\section{Proof of existence of two modes}
\begin{proof}
 Let $y \in \mathbf{Z}_\mathrm{m}$, and we define $k_1 = (\alpha + y) \bmod \mathrm{m}$ and $k_2 = (\alpha - y) \bmod \mathrm{m}$. Then, we evaluate the $\zeta$ for both $k_1$ and $k_2$:

 \begin{align*}
\zeta_{k_1} &= ((\alpha + y) - \alpha) \bmod \mathrm{m} = y \\
\zeta_{k_2} &= ((\alpha - y) - \alpha) \bmod \mathrm{m} = -y \bmod \mathrm{m} = \mathrm{m} - y
\end{align*}

Using Equation \eqref{eq:3.7} and fixing the parameter $\mu = \pi$, we evaluate the conditional pmf of $\Theta$ given $X = x$ at the corresponding angles $\theta_1 = \frac{2 \pi k_1}{\mathrm{m}}$ and $\theta_2 = \frac{2\pi k_2}{\mathrm{m}}$

\begin{align*} 
 P(\theta_1 \mid x) \propto (q^y + q^{\mathrm{m}-y}) \exp\Big(-(\beta x)^{\gamma}\Big(1 - \tanh(d) \cos\Big(\pi-\frac{2 \pi y}{\mathrm{m}}\Big)\Big)^\varepsilon\Big),
\end{align*}
 and
 \begin{align*} 
 P(\theta_2 \mid x) \propto (q^{\mathrm{m}-y} + q^{\mathrm{m}-(\mathrm{m} - y)}) \exp\Big(-(\beta x)^{\gamma}\Big(1 - \tanh(d) \cos\Big(\pi-\frac{2 \pi (\mathrm{m} - y)}{\mathrm{m}}\Big)\Big)^\varepsilon\Big).
\end{align*}

Since the cosine is an even function, $\cos(-\pi + \frac{2\pi y}{\mathrm{m}}) = \cos(\pi - \frac{2\pi y}{\mathrm{m}})$. Consequently, $P(\theta_1 \mid x) = P(\theta_2 \mid x)$, which proves the conditional distribution is perfectly symmetric about $\alpha$.

Next, we establish the existence of two modes. Retaining $\mu = \pi$ and further fixing $\varepsilon = 1$, we define a strictly positive constant $C = (\beta x)^\gamma \tanh(d)$. Using the trigonometric identity $\cos(\pi - \theta) = -\cos(\theta)$ and dropping the factor $\exp(-(\beta x)^\gamma)$ which is independent of $k$, the proportional probability at any corresponding $ \theta_k = 2 \pi k/\mathrm{m}$ simplifies to
\begin{equation*}
P(\theta_k \mid x) \propto (q^{\zeta_k} + q^{\mathrm{m}-\zeta_k}) \exp\Big(-C \cos\Big(\frac{2\pi\zeta_k}{\mathrm{m}}\Big)\Big),
\end{equation*}
where $\zeta_k = (k-\alpha) \bmod \mathrm{m}$.

To establish that the conditional distribution has a local maximum at the angle $\theta_\alpha = \frac{2 \pi \alpha}{\mathrm{m}}$, the conditional pmf at $\theta_\alpha$ must be strictly greater than the pmf at its immediate neighbors. These neighbors correspond to the angles $\theta_k = \frac{2 \pi k}{\mathrm{m}}$, where $k = (\alpha \pm 1) \bmod \mathrm{m}$. Due to the symmetry, it is sufficient to show that the pmf at $\theta_\alpha$ is strictly greater than the pmf at $\theta_{k_1} = \frac{2 \pi k_1}{\mathrm{m}}$, where $k_1 = (\alpha + 1) \bmod \mathrm{m}$. At $k = \alpha$, we have $\zeta_k = 0$, and at $k_1$, we have $\zeta_{k_1} = 1$. Then, the condition for the first mode is $P(\theta_\alpha \mid x) > P(\theta_{k_1} \mid x)$, that is
\begin{align*}
(1 + q^{\mathrm{m}}) e^{-C \cos(0)} &> (q + q^{\mathrm{m}-1}) e^{-C \cos(2\pi/\mathrm{m})}.
\end{align*}

By rearranging the terms and taking the logarithm, it yields

\begin{equation}
C \left(1 - \cos\left(\frac{2\pi}{\mathrm{m}}\right)\right) < \ln\left( \frac{1 + q^{\mathrm{m}}}{q + q^{\mathrm{m}-1}} \right)
\end{equation}

Similarly for the second mode at the angle $\theta _l = 2 \pi l / \mathrm{m}$, where $l = (\alpha + \mathrm{m}/2) \bmod \mathrm{m}$. Its neighbors are  $2 \pi (l \pm 1) / \mathrm{m}$, where $ l \pm 1 = (\alpha + \mathrm{m}/2 \pm 1) \bmod \mathrm{m}$. By symmetry, it is sufficient to evaluate the neighbor $l_1 = (\alpha + \mathrm{m}/2 - 1) \bmod \mathrm{m}$. At $l$, we have $\zeta_l = \mathrm{m}/2$, and at $l_1$, we have $\zeta_{l_1} = \mathrm{m}/2 - 1$. Then, the condition on the second mode is $P(\theta_l \mid x) > P(\theta_{l_1} \mid x)$, that is

\begin{equation*}
2q^{\mathrm{m}/2} e^C > (q^{\mathrm{m}/2-1} + q^{\mathrm{m}/2+1}) e^{C \cos(2\pi/\mathrm{m})}
\end{equation*}

By rearranging the terms and taking the logarithm, it yields

\begin{equation}
C \left(1 - \cos\left(\frac{2\pi}{\mathrm{m}}\right)\right) > \ln\left( \frac{1 + q^2}{2q} \right).
\end{equation}

From both conditions, 

\begin{equation*}
\ln\left( \frac{1 + q^{\mathrm{m}}}{q + q^{\mathrm{m}-1}} \right)  > \ln\left( \frac{1 + q^2}{2q} \right). 
\end{equation*} 
 By removing the logarithm and rearranging the term it implies

 \begin{align*}
2q(1 + q^{\mathrm{m}}) &> (1 + q^2)(q + q^{\mathrm{m}-1}) \\
 (q - q^{\mathrm{m}-1})(1 - q^2) &> 0
\end{align*}

Because $q \in (0, 1)$ and $\mathrm{m} \ge 4$, it is guaranteed that $(1 - q^2) > 0$ and $(q - q^{\mathrm{m}-1}) > 0$. The product of two strictly positive terms is positive, confirming the inequality holds. This proves that an interval for $C$ (and thus $x$) always exists where $P(\theta \mid x)$ is strictly bimodal. 
\end{proof}


\section{Proof of \autoref{theo:4.1}}
\begin{proof}
Let $\mathbf{h} = 1 - \tanh(d) \cos\!\left(\mu - \tfrac{2\pi\zeta}{\mathrm{m}}\right)$ and 
$\mathbf{g}(\gamma) = \Gamma(1 + 2\gamma) - \Gamma(1 + \gamma)^2$.  
 
If $\mu - \tfrac{2\pi\zeta}{\mathrm{m}}$ lies in the first or fourth quadrant, then $0 < \mathbf{h} \leq 1$, which implies that $\mathbf{h}^{1/\gamma}$ is an increasing function of $\gamma$. Conversely, if $\mu - \tfrac{2\pi\zeta}{\mathrm{m}}$ lies in the second or third quadrant, then $1 \leq \mathbf{h} < 2$, implying that $\mathbf{h}^{1/\gamma}$ is a decreasing function of $\gamma$.   

Let $\psi(\cdot)$ denote the digamma function, $\psi(t)=\Gamma'(t)/\Gamma(t)$.  
Differentiating $\mathbf{g}(\gamma)$ and using $\psi$ yields
\begin{align*}
\mathbf{g}'(\gamma)
&= 2\Gamma(1+2\gamma)\,\psi(1+2\gamma)
   - 2\Gamma(1+\gamma)^2\,\psi(1+\gamma) \\
&= 2\Gamma(1+\gamma)^2
   \left(
     \frac{\Gamma(1+2\gamma)}{\Gamma(1+\gamma)^2}\,\psi(1+2\gamma)
     - \psi(1+\gamma)
   \right).
\end{align*}

By log–convexity of the Gamma function we have $\Gamma(1+\gamma)^2 \le \Gamma(1+2\gamma)$, and since the digamma function $\psi$ is increasing (see Chapter~2 of \cite{artin2015}), the term in parentheses is strictly positive. Hence $\mathbf{g}'(\gamma)>0$, which implies that $\mathbf{g}(1/\gamma)$ is a decreasing function of $\gamma$.

We now combine these monotonicity properties.

If $\mu - \tfrac{2\pi\zeta}{m}$ lies in the first or fourth quadrant and $\varepsilon = 1$, then in both $\operatorname{E}(X \mid \Theta=\theta)$ and $\operatorname{Var}(X \mid \Theta=\theta)$ the numerators decrease and the denominators increase with $\gamma$. Hence, both conditional moments decrease in $\gamma$.

If $\mu - \tfrac{2\pi\zeta}{m}$ lies in the second or third quadrant and $\varepsilon = -1$, then in each expression the relevant components are decreasing functions of $\gamma$, implying that $\operatorname{E}(X \mid \Theta=\theta)$ and $\operatorname{Var}(X \mid \Theta=\theta)$ are decreasing in $\gamma$.
\end{proof}


\section{Proof of \autoref{theo.5}}
\begin{proof}
Assume $d=0$. In this case, the first-order moments are

\begin{align*}
\operatorname{E}(X\cos \Theta) &= \frac{(1-q)^2}{\beta A}\,\Gamma\!\Big(1+\tfrac{1}{\gamma}\Big)\,\cos \tfrac{2\pi\alpha}{\mathrm{m}},\\
\operatorname{E}(X\sin \Theta) &= \frac{(1-q)^2}{\beta A}\,\Gamma\!\Big(1+\tfrac{1}{\gamma}\Big)\,\sin\tfrac{2\pi\alpha}{\mathrm{m}},\\
\operatorname{E}(\cos \Theta) &= \frac{(1-q)^2}{A}\,\cos \tfrac{2\pi\alpha}{\mathrm{m}},\qquad
\operatorname{E}(\sin \Theta) = \frac{(1-q)^2}{A}\,\sin \tfrac{2\pi\alpha}{\mathrm{m}},\\
\operatorname{E}(X) &= \frac{1}{\beta^{2}}\,\Gamma\!\Big(1+\tfrac{2}{\gamma}\Big),
\end{align*}

where $A = 1-2q\cos\tfrac{2\pi}{\mathrm{m}}+q^2$.
These identities imply that
\[
\operatorname{cov}(X,\cos\Theta)=\operatorname{cov}(X,\sin\Theta)=0.
\]
In the nondegenerate case (i.e., $\operatorname{var}(X)$, $\operatorname{var}(\cos\Theta)$, and $\operatorname{var}(\sin\Theta)$ are nonzero), it follows that
\[
r_{xc}=\operatorname{corr}(X,\cos\Theta)=0
\quad\text{and}\quad
r_{xs}=\operatorname{corr}(X,\sin\Theta)=0.
\]

Since $|r_{cs}| < 1$ in the nondegenerate setting (i.e., $\cos\Theta$ and $\sin\Theta$ are not perfectly correlated), substituting $r_{xc} = r_{xs} = 0$ into the angular--radial correlation formula \eqref{eq:4.1} yields  
$
\rho^2_{x\theta} = 0.
$
\end{proof}

\end{appendices}


\end{document}